\documentclass{article}
\usepackage{amsfonts}
\usepackage{amssymb}

\newcommand{\sC}{{\mathbb C}}

\newcommand{\sH}{{\mathbb H}}

\newcommand{\sN}{{\mathbb N}}

\newcommand{\sR}{{\mathbb R}}

\newcommand{\sZ}{{\mathbb Z}}

\newcommand{\cD}{{\mathcal D}}

\newcommand{\cH}{{\mathcal H}}

\newcommand{\cL}{{\mathcal L}}

\newcommand{\cR}{{\mathcal R}}

\newcommand{\ind}{{\rm ind} \,}
\newcommand{\op}{{\rm op} \,}

\newcommand{\spec}{{\rm sp} \,}
\newcommand{\spd}{{\rm sp}_{dis} \,}
\newcommand{\spe}{{\rm sp}_{ess} \,}

\newcommand{\rot}{{\rm rot} \,}
\newcommand{\grad}{{\rm grad} \,}
\newtheorem{theorem}{Theorem}

\newtheorem{corollary}[theorem]{Corollary}

\newtheorem{proposition}[theorem]{Proposition}

\newenvironment{proof}[1][Proof]{\noindent\textbf{#1.} }{\hfill \rule{0.5em}{0.5em}}
\begin{document}
\title{Agmon's type estimates of exponential behavior of solutions of
systems of elliptic partial differential equations.\\
Applications to Schr\"{o}dinger, Moisil-Theodorescu and Dirac
operators.}
\author{V. Rabinovich\thanks{Partially supported by the DFG
grant 444 MEX-112/2/05.}, S. Roch \\
Vladimir Rabinovich, Instituto Polit\'{e}cnico Nacional, \\
ESIME-Zacatenco, Av. IPN, edif.1, M\'{e}xico D.F., \ \\
07738, M\'{E}XICO, \\e-mail: vladimir.rabinovich@gmail.com\\
Steffen Roch, Technische Universit\"{a}t Darmstadt, \\
Schlossgartenstrasse 7, 64289 Darmstadt, Germany,\\
e-mail: roch@mathematik.tu-darmstadt.de }
\date{}
\maketitle
\begin{abstract}
The aim of this paper is to derive Agmon's type exponential
estimates for solutions of elliptic systems of partial
differential equations on $\sR^n$. We show that these estimates
are related with the essential spectra of a family of associated
differential operators which depend on the original operator, and
with exponential weights which describe the decrease of solutions
at infinity. The essential spectra of the involved operators are
described by means of their limit operators.

The obtained results are applied to study the problem of
exponential decay of eigenfunctions of matrix Schr\"{o}dinger,
Moisil-Theodorescu, and Dirac operators.
\end{abstract}
\section{Introduction}
The main aim of the paper is to obtain Agmon's type exponential
estimates for the decaying behavior of solutions of systems
of elliptic partial differential equations with variable
coefficients. Exponential decay estimates of this type are intensively
studied in the literature. We only mention Agmon's pioneering papers
\cite{Ag1,Ag} where estimates of eigenfunctions of second order
elliptic operators were obtained in terms of a special metric, now
called the Agmon metric. Exponential estimates for solutions of
pseudodifferential equations on $\sR^n$ are also considered in
\cite{LR,Man,Mar,Mar1,Nakam,RTUB,Z,CM}.

In this paper, we propose a new approach to exponential estimates for
solutions of systems of partial differential equations. Our approach is
based on the limit operators method. This method was employed earlier
to study the essential spectrum of perturbed pseudodifferential
operators, which has found applications to electro-magnetic Schr\"{o}dinger
operators, square-root Klein-Gordon, and Dirac operators under very
general assumptions for the magnetic and electric potentials at infinity.
Based on the limit operators method, a simple and transparent
proof of the well known Hunziker, van Winter, Zjislin theorem (HWZ-Theorem)
for multi-particle Hamiltonians was derived in \cite{CM,RJMP}. In
\cite{RRJP}, the limit operators method was applied to study the
essential spectrum of discrete Schr\"{o}dinger operators.

Let
\[
A(x,D)u(x) = \sum_{|\alpha| \le m} a_\alpha(x) D^\alpha u(x), \quad
x \in \sR^n
\]
be a uniformly elliptic system of partial differential operators
of order $m$ on $\sR^n$ with bounded and uniformly continuous
$N \times N$ matrix-valued coefficients $a_\alpha$. Further let $w
:= \exp v$ be a weight on $\sR^n$ such that $\lim_{x \to \infty} w(x)
= \infty$ and assume that all first and higher order derivatives of
$v$ exist and are bounded.

With the operator $A(x,D)$ and the weight $w$, we associate a family 
of partial differential operators,
\[
B_t := A(x,D + it \nabla v(x)), \quad t \in [-1, 1].
\]
We let $\spe B_t$ denote the essential spectrum of $B_t$ considered
as an unbounded closed operator on the Hilbert space $L^2(\sR^n,\sC^N)$
with domain $H^m(\sR^n, \sC^N)$. Further we let $L^2(\sR^n,\sC^N,w)$ and
$H^m(\sR^n,\sC^N,w)$ refer to the corresponding weighted spaces of
vector-valued functions $u$ such that $wu \in L^2(\sR^n, \sC^N)$ and
$wu \in H^m(\sR^n, \sC^N)$, respectively. The following result, which
will be proved in this paper, provides exponential estimates for the 
decay of the solutions of the equation $A(x,D)u = f$.
\begin{theorem} \label{t0.1}
Let $A(x,D)$ be a uniformly elliptic on $\sR^n$ matrix partial
differential operator with bounded and uniformly continuous
coefficients, and assume that $0 \notin \spe B_t$ for every $t \in
[-1, 1]$. Let $u \in H^m(\sR^n, \sC^N, w^{-1})$ be a solution of the
equation $A(x,D)u = f$ with $f \in L^2(\sR^n,\sC^N,w)$. Then $u \in
H^m(\sR^n,\sC^N,w)$.
\end{theorem}
By this result, the derivation of exponential estimates of solutions
of the equation $A(x,D)u=f$ is basically reduced to the calculation
of the essential spectrum of the $B_t$. The latter can be done by
means of the limit operators method (see \cite{RRS1,RRS2} and
the monograph \cite{RRSB}). For, one associates with each operator
$B_t$ the family $\op (B_t)$ of its limit operators $B_t^g$ which,
roughly speaking, describe the behaviour of the operator at infinity. 
Then it follows from \cite{RRS1,RRS2,RRSB} that
\begin{equation} \label{0.1}
\spe B_t = \bigcup \limits_{B_t^g \in \op (B_t)} \spec B_t^g,
\end{equation}
where $\spec B_t^g$ denotes the spectrum of the unbounded operator
$B_t^g$ acting on $L^2(\sR^n,\sC^N)$.

In many important instances, the limit operators are of an enough
simple structure, namely partial differential operators with constant 
coefficients. Then formula (\ref{0.1}) provides an effective tool 
to calculate the essential spectra of partial differential operators 
and thus, in view of Theorem \ref{t0.1}, to get exponential estimates 
for the solutions of the equation $A(x,D)u = f$. We will illustrate 
this statement by applying Theorem \ref{t0.1} to verify the exponential 
decay and to obtain explicit exponential estimates for the eigenvectors 
of Schr\"{o}dinger operators with matrix potentials, of 
Moisil-Theodorescu quaternionic operators with variable coefficients, 
and of Dirac operators.

The contents of the paper is as follows. In Section 2 we introduce
the main definitions and prove Theorem \ref{t0.1}. The following
sections are devoted to several applications of Theorem \ref{t0.1}.
We start in Section 3 with the essential spectrum and exponential
estimates for eigenvectors of Schr\"{o}dinger operators with matrix
potentials. Note that the well known Pauli operator is a Schr\"{o}dinger
operator of this kind (see, for instance, \cite{Cycon}). Moreover,
such operators appear in the Born-Oppenheimer approximation for
polyatomic molecules \cite{KM,Mar00,Nedelec}. For potentials
which are slowly oscillating at infinity we describe the location
of the essential spectrum and give exact estimates of the behavior of
eigenfunctions of the discrete spectrum at infinity.

In Section 4 we consider general Moisil-Theodorescu (quaternionic)
operators with variable coefficients. Note that numerous important
systems of partial differential operators of quantum mechanics,
elasticity theory, and field theory admit a formulation in terms
of quaternionic operators (see \cite{Le,GS,K1,KS,Scol} and the
references cited there). We shall verify explicit necessary and
sufficient conditions for quaternionic operators with variable
coefficients to be Fredholm operators and derive exponential
estimates at infinity for solutions of Fredholm quaternionic equations.

In the concluding Section 5 we consider the Dirac operator on $\sR^3$,
equipped with a Riemannian metric, with electric and magnetic potentials
which are slowly oscillating at infinity.
\section{Exponential estimates  of solutions of systems of partial
differential equations}
\subsection{Essential spectrum}
We will use the following standard notations.
\begin{itemize}
\item
Given Banach spaces $X, Y$, $\cL(X,Y)$ is the space of all bounded
linear operators from $X$ into $Y$. We abbreviate $\cL (X,X)$ to
$\cL(X)$.
\item
$L^2(\sR^n, \sC^N)$ is the Hilbert space of all measurable
functions on $\sR^n$ with values in $\sC^N$, provided with the
norm
\[
\|u\|_{L^2(\sR^n,\sC^N)} :=
\left( \int_{\sR^n} \|u(x)\|_{\sC^N}^2 dx \right)^{1/2}.
\]
\item
The unitary operator $V_h$ of shift by $h \in \sR^n$ acts on
$L^2(\sR^n,\sC^N)$ via $(V_h u)(x) := u(x-h)$.
\item
$C_b (\sR^n)$ is the $C^*$-algebra of all bounded continuous
functions on $\sR^n$.
\item
$C_b^u (\sR^n)$ is the $C^*$-subalgebra of $C_b(\sR^n)$ of all
uniformly continuous functions.
\item
$SO(\sR^n)$ is the $C^*$-subalgebra of $C_b^u (\sR^n)$ which consists
of all functions $a$ which are slowly oscillating in the sense that
\[
\lim_{x \to \infty} \sup_{y \in K} |a(x+y) - a(x)| = 0
\]
for every compact subset $K$ of $\sR^n$.
\item
$SO^1(\sR^n)$ is the set of all bounded differentiable functions
$a$ on $\sR^n$ such that
\[
\lim_{x \to \infty} \frac{\partial a(x)}{\partial x_{j}} = 0 \;
\mbox{for} \; j = 1, \, \ldots, \, n.
\]
Evidently, $SO^1(\sR^n) \subset SO(\sR^n)$.
\end{itemize}
We also use the standard multi-index notation. Thus, $\alpha =
(\alpha_1,...,\alpha_n)$ with $\alpha_j \in \sN \cup \{0\}$ is a
multi-index, $|\alpha| = \alpha_1 + \ldots + \alpha_n$ is its
length, and
\[
\partial^\alpha := \partial_{x_1}^{\alpha_1} \ldots
\partial_{x_n}^{\alpha_n} \quad \mbox{and} \quad
D^\alpha := (-i\partial_{x_1})^{\alpha_1} \ldots
(-i\partial_{x_n})^{\alpha_n}
\]
are the operators of $\alpha^{th}$ derivative. Finally, $\langle 
\xi \rangle := (1 + |\xi|^2)^{1/2}$ for $\xi \in \sR^n$.

We consider matrix partial differential operators of order $m$ of
the form
\begin{equation} \label{e1.1}
(Au)(x) = \sum_{|\alpha| \le m} a_\alpha (x) (D^\alpha u)(x),
\quad x \in \sR^n
\end{equation}
under the assumption that the coefficients $a_\alpha$ belong to
the space
\[
C_b^u (\sR^n, \cL(\sC^N)) := C_b^u(\sR^n) \otimes \cL(\sC^N).
\]
The operator $A$ in $(\ref{e1.1})$ is considered as a bounded linear
operator from the Sobolev space $H^m (\sR^n, \sC^n)$ to
$L^2(\sR^n,\sC^N)$. The operator $A$ is said to be \emph{uniformly
elliptic} on $\sR^n$ if
\begin{equation} \label{1.2}
\inf_{x\in \sR^n, \, \omega \in S^{n-1}} \left| \det
\sum_{|\alpha| = m}^N a_\alpha (x) \omega^\alpha \right\vert > 0
\end{equation}
where $S^{n-1}$ refers to the unit sphere in $\sR^n$.

The Fredholm properties of the operator $A$ can be expressed in
terms of its limit operators which are defined as follows. Let $h
: \sN \to \sR^n$ be a sequence which tends to infinity. The
Arcel\`a-Ascoli theorem combined with a Cantor diagonal argument
implies that there exists a subsequence $g$ of $h$ such that the
sequences of the functions $x \mapsto a_\alpha (x + g(k))$
converges as $k \to \infty$ to a limit function $a_\alpha^g$
uniformly on every compact set $K \subset \sR^n$ for every
multi-index $\alpha$. The operator
\[
A^g := \sum_{|\alpha| \le m} a_\alpha^g D^\alpha
\]
is called the \emph{limit operator of $A$ defined by the sequence}
$g$. Equivalently, $A^g$ is the limit operator of $A$ with respect
to $g$ if and only if, for every function $\chi \in C_0^\infty
(\sR^n)$,
\[
\lim_{m \to \infty} V_{-g(m)} A V_{g(m)} \chi I = A^g \chi I
\]
in the space $\cL (H^m(\sR^n, \sC^N), L^2(\sR^n, \sC^N))$ and
\[
\lim_{m \to \infty} V_{-g(m)} A^* V_{g(m)} \chi I = (A^g)^* \chi I
\]
in $\cL (L^2(\sR^n, \sC^N), H^{-m} (\sR^n, \sC^N))$. Here,
\[
A^* u = \sum_{|\alpha| \le m} D^\alpha (a_\alpha^* u)
\quad \mbox{and} \quad
(A^g)^* u = \sum_{|\alpha| \le m} D^\alpha ((a_\alpha^g)^* u)
\]
refer to the adjoint operators of $A, \, A_g : H^m (\sR^n, \sC^N) \to
L^2(\sR^n,\sC^N)$. Let finally $\op (A)$ denote the set of all limit
operators of $A$ obtained in this way.
\begin{theorem} \label{t1.1}
Let $A$ be a uniformly elliptic differential operator of the form
$(\ref{e1.1})$. Then $A : H^m(\sR^n,\sC^N) \to L^2(\sR^n,\sC^N)$ is
a Fredholm operator if and only if all limit operators of $A$ are
invertible as operators from $H^m(\sR^n,\sC^N)$ to $L^2(\sR^n,\sC^N)$.
\end{theorem}
\begin{proof}
Let $(I - \Delta)^\alpha$ be the pseudodifferential operator with
symbol $(1 + |\xi|^2)^\alpha$, and let $A$ be a partial
differential operator of the form (\ref{e1.1}). Then the operator
$B := A (I - \Delta)^{-m/2}$ belongs to the matrix Wiener algebra
considered in \cite{RR}. It follows from results of \cite{RR} that
$B$, considered as an operator from $L^2(\sR^n,\sC^N)$ to
$L^2(\sR^n,\sC^N)$, is a Fredholm operator if and only if all
limit operators of $B$ are invertible on $L^2(\sR^n,\sC^N)$. Since
the limit operators $B_g$ of $B$ and $A_g$ of $A$ are related by
the relation $B_g = A_g (1 - \Delta)^{-m/2}$ (which comes from the 
shift invariance of the operator $\Delta$), the assertion follows.
\end{proof} \\[3mm]
The uniform ellipticity of the operator $A$ implies the a priori
estimate
\begin{equation} \label{1.3}
\|u\|_{H^m(\sR^n,\sC^N)} \le C \left( \|Au\|_{L^2(\sR^n,\sC^N)} +
\|u\|_{L^2(\sR^n,\sC^N)} \right).
\end{equation}
This estimate allows one to consider the uniformly elliptic
differential operator $A$ as a closed unbounded operator on
$L^2(\sR^n,\sC^N)$ with dense domain $H^m (\sR^n,\sC^N)$. It turns
out (see \cite{Agran}) that $A$, considered as an unbounded
operator in this way, is a (unbounded) Fredholm operator if and
only if $A$, considered as a bounded operator from $H^m(\sR^n,
\sC^N)$ to $L^2(\sR^n, \sC^N)$, is a (common bounded) Fredholm
operator.

We say that $\lambda \in \sC$ belongs to the \emph{essential
spectrum} of $A$ if the operator $A - \lambda I$ is not Fredholm
as an unbounded differential operator. As above, we denote the 
essential spectrum of $A$ by $\spe A$ and the common spectrum of 
$A$ (considered as as unbounded operator) by $\spec A$. Then the 
assertion of Theorem \ref{t1.1} can be stated as follows.
\begin{theorem} \label{t1.2}
Let $A$ be a uniformly elliptic differential operator of the form
$(\ref{e1.1})$. Then
\begin{equation} \label{1.4}
\spe A = \bigcup\limits_{A_g \in \op A} \spec A_g.
\end{equation}
\end{theorem}
\subsection{Exponential estimates}
Let $w$ be a positive measurable function on $\sR^n$, which we
call a weight. By $L^2(\sR^n,\sC^N,w)$ we denote the space of all
measurable functions on $\sR^n$ such that
\[
\|u\|_{L^2(\sR^n,\sC^N,w)} := \|wu\|_{L^2(\sR^n,\sC^N)} < \infty.
\]
In what follows we consider weights of the form $w = \exp v$ where
$\partial_{x_j} v \in C_b^\infty (\sR^n)$ for $j = 1, \ldots, n$
and
\begin{equation} \label{1.5}
\lim_{x \to \infty} \partial_{x_i x_j}^2 v(x) = 0 \quad \mbox{for}
\; 1 \le i,j \le n.
\end{equation}
We call weights with these properties \emph{slowly oscillating}
and let $\cR$ stand for the class of all slowly oscillating
weights.

Examples of slowly oscillating weights can be constructed as
follows. Given a positive $C^\infty$-function $l : S^{n-1} \to
\sR$, set $v_l(x) := l(x/|x|) |x|$. Then $w_l := \exp v_l$ defines
a weight on $\sR^n$. Clearly, $v_l$ is a positively homogeneous
function, that is $v_l(tx) = tv_l(x)$ for all $t>0$ and $x \in
\sR^n$. Moreover, $v_l \in C^\infty (\sR^n \setminus \{0 \})$, and
$\nabla v_l(\omega) = l(\omega) \omega$ for every point $\omega
\in S^{n-1}$. Let $\tilde{v}_l$ refer to a $C^\infty$-function on
$\sR^n$ which coincides with $v_l$ outside a small neighborhood of
the origin. Then the weight $\tilde{w}_l := \exp \tilde{v}_l$
belongs to the class $\cR$. Moreover,
\begin{equation} \label{1.7}
\lim_{x \to \eta_\omega} \nabla \tilde{v}_l(x) = \nabla
v_l(\omega) = l(\omega) \omega
\end{equation}
for $\omega \in S^{n-1}$.
\begin{proposition} \label{p1.1}
Let $A$ be a differential operator of the form $(\ref{e1.1})$, and
let $w = \exp v$ be a weight in $\cR$. Then
\[
w^{-1}Aw = \sum_{|\alpha| \le m} a_\alpha (D + i\nabla v)^\alpha +
\Phi + R
\]
where $R := \sum_{|\alpha| \le m-1} b_\alpha D^\alpha$ is a
differential operator with continuous coefficients such that
$\lim_{x \to \infty} b_\alpha (x) = 0$.
\end{proposition}
For a proof see \cite{Z} where a similar result is derived for a
large class of pseudo-differential operators. The following is
taken from \cite{GF}, p. 308.
\begin{proposition} \label{p1.2}
Let $X_1, X_2, Y_1$ and $Y_2$ be Banach spaces such that $X_1$ is
densely embedded into $X_2$ and $Y_1$ is embedded into $Y_2$.
Further let $A : X_2 \to Y_2$ and $A|_{X_1} : X_1 \to Y_1$ be
Fredholm operators, and suppose that
\[
\ind (A : X_2 \to Y_2) = \ind (A|_{X_1} : X_1 \to Y_1).
\]
If $u \in X_2$ is a solution of the equation $Au = f$ with
right-hand side $f \in Y_1$, then $u \in X_1$.
\end{proposition}
\begin{theorem} \label{t1.3}
Let $A$ be a uniformly elliptic differential operator of the form
$(\ref{e1.1})$, and let $w =\exp v$ be a weight in $\cR$ such that
$\lim_{x \to \infty} w(x) = +\infty$. For $t \in [-1,1]$, set
\[
A_{w,t} := \sum_{|\alpha| \le m} a_\alpha (D + it\nabla v)^\alpha,
\]
and assume that
\begin{equation} \label{1.8}
0 \notin \bigcup \limits_{t \in [-1,1]} \spe A_{w,t} = \bigcup
\limits_{t \in [-1,1]} \bigcup \limits_{A_{w,t}^g \in \op (A_{w,t})}
\spec A_{w,t}^g.
\end{equation}
If $u$ is a function in $H^m(\sR^n, \sC^N, w^{-1})$ for which $Au$
is in $L^2(\sR^n, \sC^N, w)$, then $u$ already belongs to 
$H^m(\sR^n, \sC^N, w)$.
\end{theorem}
\begin{proof}
Note that $A : H^m(\sR^n, \sC^N, w^t) \to L^2(\sR^n, \sC^N, w^t)$
is a Fredholm operator if and only if $w^{-t}Aw^t :
H^m(\sR^n,\sC^N) \to L^2(\sR^n,\sC^N)$ is a Fredholm operator, and
that the Fredholm indices of these operator coincide. Proposition
\ref{p1.1} implies that
\[
w^{-t}Aw^t = \sum_{|\alpha| \le m} a_\alpha (D+it\nabla v)^\alpha
+ R_t,
\]
where $R_t = \sum_{|\alpha| \le m-1} b_{\alpha, t} D^\alpha$ and
$\lim_{x \to \infty} b_{\alpha, t}(x) = 0$ for every $t \in
[-1,1]$. Hence, $\op (w^{-t}Aw^t) = \op (A_{w,t})$. Moreover, it
is not hard to see that the coefficients $b_{\alpha, t}$ depend
continuously on $t \in [-1, 1]$. Hence, the family $w^{-t}Aw^t$
depends continuously on $t$, and condition (\ref{1.8}) implies
that all operators $w^{-t}Aw^t : H^m(\sR^n,\sC^N) \to
L^2(\sR^n,\sC^N)$ with $t \in [-1,1]$ are Fredholm operators and
that the Fredholm indices of these operators coincide. This
implies that each operator
\[
A : H^m(\sR^n, \sC^N, w^t) \to L^2(\sR^n, \sC^N, w^t)
\]
owns the Fredholm property and that the index of this operator is
independent of $t \in [-1, 1]$. Since $H^m(\sR^n, \sC^N, w)$ is
densely embedded into $H^m(\sR^n, \sC^N, w^{-1})$, we can apply
Proposition \ref{p1.2} to obtain that all solutions of the
equation $Au = f$ with right-hand side $f \in L^2(\sR^n,\sC^n,w)$,
which a priori are in $H^m(\sR^n, \sC^N, w^{-1})$, in fact belong
to $H^m(\sR^n, \sC^N, w)$.
\end{proof}

\begin{corollary} \label{t1.4}
Let $A$ be a uniformly elliptic differential operator of the form
$(\ref{e1.1})$, and let $w = \exp v$ be a weight in $\cR$ with
$\lim_{x \to \infty} w(x) = +\infty$. Let $\lambda \in \spd A$ and
$\lambda \notin \spe A_{tw}$ for all $t \in [0, 1]$. Then every
eigenfunction of $A$ associated with $\lambda$ belongs to the
space $H^m(\sR^n, \sC^N, w)$.
\end{corollary}

Indeed, this is an immediate consequence of Theorem \ref{t1.3}
since eigenfunctions of uniformly elliptic operator of order $m$
necessarily belong to $H^m(\sR^n,\sC^N)$.
\section{Schr\"{o}dinger operators with matrix potentials}
\subsection{Essential spectrum}
We consider the Schr\"{o}dinger operator
\begin{equation} \label{e2.1}
\cH := (i \partial_{x_j} - a_j) \rho^{jk} (i \partial_{x_k} -
a_k)E + \Phi
\end{equation}
where $E$ is the $N \times N$ unit matrix, $a=(a_1,...,a_n)$ is
referred to as the the magnetic potential, and $\Phi =
(\Phi_{pq})_{p,q = 1}^N$ is a matrix potential on $\sR^n$, the
latter equipped with a Riemann metric $\rho
=(\rho_{jk})_{j,k=1}^n$ which is subject to the positivity
condition
\begin{equation} \label{e2.2}
\inf_{x \in \sR^n, \, \omega \in S^{n-1}} \rho_{jk}(x) \omega^j
\omega^k > 0,
\end{equation}
where $\rho _{jk}(x)$ refers to the matrix inverse to $\rho
^{jk}(x)$. Here and in what follows, we make use of Einstein's
summation convention.

In what follows we suppose that $\rho^{jk}$ and $a_j$ are
real-valued functions in $SO^1(\sR^n)$ and that $\Phi_{pq} \in
SO(\sR^n)$. Under these conditions, $\cH$ can be considered as a
closed unbounded operator on $L^2(\sR^n,\sC^N)$ with domain
$H^2(\sR^n,\sC^N)$. If $\Phi$ is a Hermitian matrix-valued
function, then $\cH$ is a self-adjoint operator.

The limit operators of $\cH$ are the operators with constant
coefficients
\[
\cH^g = (i \partial_{x_j} - a_j^g) \rho_g^{jk} (i \partial_{x_k} -
a_k^g)E + \Phi^g
\]
where
\begin{equation} \label{e2.3}
a^g := \lim_{m \to \infty} a(g_m), \quad \rho_g := \lim_{m \to
\infty} \rho (g_m), \quad \Phi^g := \lim_{m \to \infty} \Phi
(g_m).
\end{equation}
The operator $\cH^g$ is unitarily equivalent to the operator
\[
\cH_1^g := -\rho_g^{jk} \partial_{x_j} \partial_{x_k}E + \Phi ^g,
\]
which on its hand is unitarily equivalent to the operator
$\overbrace{\cH_1^g}$ of multiplication by the matrix-function
\[
\overbrace{\cH_1^g}(\xi) := (\rho_g^{jk} \xi_j \xi_k)E + \Phi ^g
\]
acting on $L^2(\sR^n,\sC^N)$. Evidently,
\[
\spec \overbrace{\cH_1^g} = \bigcup\limits_{j=1}^N \Gamma_j^g
\]
where $\Gamma_j^g := \mu_j^g + \sR$ and the $\mu_j^g$, $1 \le j \le N$, 
run through the eigenvalues of the matrix $\Phi^g$. Thus, specifying 
(\ref{1.4}) to the present context we obtain the following.
\begin{theorem} \label{te2.1}
The essential spectrum of the Schr\"{o}dinger operator $\cH$ is
given by
\begin{equation} \label{e2.4}
\spe \cH = \bigcup\limits_g \bigcup\limits_{j=1}^N \Gamma_j^g
\end{equation}
where the union is taken with respect to all sequences $g$ for
which the limits in $(\ref{e2.3})$ exist.
\end{theorem}
The description (\ref{e2.4}) of the essential spectrum becomes much 
simpler if $\Phi$ is a Hermitian matrix function, in which case $\cH$ 
is a self-adjoint operator. 
\begin{theorem} \label{te2.2}
Let the potential $\Phi$ be a Hermitian and slowly oscillating
matrix function. Then
\[
\spe \cH = [d_\Phi , + \infty)
\]
where
\[
d_\Phi := \liminf_{x \to \infty} \inf_{\|\varphi\| = 1} \langle
\Phi (x) \varphi, \varphi \rangle.
\]
\end{theorem}
\begin{proof}
Since $\Phi^g$ is Hermitian matrix,
\[
\gamma _g := \inf_{\|\varphi\| = 1} \langle \Phi^g \varphi,
\varphi \rangle
\]
is the smallest eigenvalue of $\Phi^g$. Hence, $\spec \cH^g = 
[\gamma_g, + \infty)$ and, according to (\ref{1.4}),
\[
\spe \cH = \bigcup\limits_g \, [\gamma_g, + \infty) = [\inf_g
\gamma _g, + \infty).
\]
It remains to show that
\begin{equation} \label{e2.5}
\inf_g \gamma_g = d_\Phi.
\end{equation}
Let $g$ be a sequence tending to infinity for which the limit
\[
\Phi^g := \lim_{m \to \infty} \Phi(g(m))
\]
exists. Then, for each unit vector $\varphi \in \sC^N$,
\[
\langle \Phi^g \varphi, \varphi \rangle = \lim_{m \to \infty}
\langle \Phi(g(m)) \varphi, \varphi \rangle \ge \liminf_{x \to
\infty} \langle \Phi(g(m)) \varphi, \varphi \rangle \ge d_\Phi,
\]
whence $\gamma_g \ge d_\Phi$. For the reverse inequality, note
that there exist a sequence $g_0$ tending to infinity and a
sequence $\varphi$ in the unit sphere in $\sC^N$ with limit
$\varphi_0$ such that
\[
d_\Phi = \lim_{m \to \infty} (\Phi(g_0(m)) \varphi_m, \varphi_m) =
(\Phi^{g_0} \varphi_0, \varphi_0) \ge \gamma_{g_0}.
\]%
Thus, $\gamma_{g_0} = d_\Phi$, whence (\ref{e2.5}).
\end{proof}
\subsection{Exponential estimates of eigenfunctions of the
discrete spectrum}
Here we suppose that the components $\rho^{jk}$ of the Riemann
metric, the coefficients $a_\alpha$ and the weight $w$ are slowly
oscillating functions and that $\Phi$ is a Hermitian slowly
oscillating matrix function. Every limit operator $(w^{-1} \cH
w)_g$ of $w^{-1} \cH w$ is unitarily equivalent to the operator
\[
\cH_w^g := \rho_g^{jk} (D_{x_j} + i (\nabla v)_j^g) (D_{x_k} + i
(\nabla v)_k^g) + \Phi^g E,
\]
where $\rho_g^{jk}$ and $\Phi^g$ are the limits defined by
(\ref{e2.3}) and
\begin{equation} \label{5.6}
(\nabla v)^g := \lim_{k \to \infty} (\nabla v)(g(k)) \in \sR^n.
\end{equation}
We set
\[
|(\nabla v)|_\rho^2 := \rho^{jk} (\nabla v)_j (\nabla v)_k \quad
\mbox{and} \quad
|(\nabla v)^g|_{\rho_g}^2 := \rho_g^{jk} (\nabla
v)_j^g (\nabla v)_k^g.
\]
The operator $\cH_w^g$ is unitarily equivalent to the operator of
multiplication by the matrix-valued function
\[
\overbrace{\cH_w^g} (\xi) := \rho_g^{jk} (\xi_j + i (\nabla
v)_j^g) (\xi_k + i (\nabla v)_k^g) + \Phi^g E, \quad \xi \in \sR^n,
\]
the real part of which is
\begin{equation} \label{5.5}
\mathfrak{R}(\overbrace{\cH_w^g}) = \rho_g^{jk} \xi_j \xi_k +
(\Phi^g - |(\nabla v)^g|_{\rho_g}^2 E).
\end{equation}
Corollary \ref{t1.4} implies the following.
\begin{theorem} \label{te2.3}
Let $\lambda \in \spd \cH$, and let $w =\exp v$ be a weight in
$\cR$ for which
\[
\limsup_{x \to \infty} |(\nabla v) (x)|_{\rho (x)} < \sqrt{d_\Phi
- \lambda}.
\]
Then every $\lambda$-eigenfunction of $\cH$ belongs to $H^2(\sR^n,
w)$.
\end{theorem}
\begin{corollary} \label{c2.1}
Let $\lambda \in \spd \cH$, and let $c \in \sR$ satisfy
\[
0 < c < \frac{\sqrt{d_\Phi - \lambda}}{\rho^{\sup}}
\]
where
\[
\rho^{\sup} := \liminf_{x \to \infty} \sup_{\omega \in S^{n-1}}
(\rho^{jk}(x) \omega_j \omega_k)^{1/2}.
\]
Then the every $\lambda$-eigenfunction of $\cH$ belongs to the
space $H^2(\sR^n, \sC^N, w)$ with weight $w(x) = e^{c \langle x
\rangle}$.
\end{corollary}
\section{Quaternionic operators}
We let $\sH (\sC)$ denote the complex quaternionic algebra, which
is the associative algebra over the field $\sC$ of complex numbers
generated by four elements $1, e_1, e_2, e_3$ subject to the
conditions
\[
e_1 e_2 = e_3, \quad e_2 e_3 = e_1, \quad e_3 e_1 = e_2
\]
and
\[
1^2 =1, \quad e_k^2 = -1, \quad 1e_k = e_k1 = e_k, \quad e_j e_k =
-e_k e_j
\]
for $j, k = 1, 2, 3$. Each of the elements $1, e_1, e_2, e_3$
commutes with the imaginary unit $i$. Hence, every element $q \in \sH
(\sC)$ has a unique decomposition
\[
q = q_0 + q_1e_1 + q_2e_2 + q_3e_3 =: q_0 + \mathbf{q}
\]
with complex numbers $q_j$. The number $q_0$ is called the scalar
part of the quaternion $q$, and $\mathbf{q}$ is its vector part.
One can also think of $\sH (\sC)$ as a complex linear space of
dimension four with usual linear operations. With respect to the
base $\{ 1, e_1, e_2, e_3 \} $ of this space, the operator of
multiplication by $1$ has the unit matrix $E_4$ as its matrix
representation, whereas the matrix representations $\gamma_j$ of
the operators of multiplication by $e_j$, $j=1, 2,3$, are real and
skew-symmetric, that is $\gamma_j^t = - \gamma_j$. The space $\sH
(\sC)$ carries also the structure of a complex Hilbert space via
the scalar product
\[
\langle q, r \rangle_{\sH (\sC)} :=  q_0 \overline{r_0} + q_1
\overline{r_1} + q_2 \overline{r_2} + q_3 \overline{r_3}.
\]
By $L^2 (\sR^3, \sH (\sC))$ we denote the Hilbert space of all
measurable and squared integrable quaternion-valued functions on
$\sR^3$ which is provided with the scalar product
\[
\langle u, v\rangle_{L^2 (\sR^3, \sH (\sC))} := \int_{\sR^3}
\langle u(x), v(x) \rangle_{\sH (\sC)} dx.
\]
In a similar way, we introduce the quaternionic Sobolev space $H^1
(\sR^3, \sH (\sC))$. Further we write $M^\varphi$ for the operator
of multiplication from the right by the complex quaternionic
function $\varphi$, that is
\[
(M^\varphi u)(x) = u(x) \varphi(x) \quad \mbox{for} \; x \in
\sR^3.
\]
Clearly, if $\varphi \in L^\infty (\sR^3, \sH (\sC))$, then
$M^\varphi$ acts as a bounded linear operator on $L^2 (\sR^3, \sH
(\sC))$.

Differential operators of the form
\begin{equation} \label{3.1}
A(x, D) u(x) : =\sum_{j=1}^3 a_j(x) D_{x_j} e_j u(x) + M^{\varphi
(x)}u(x), \quad x \in \sR^3,
\end{equation}
can be considered as generalized Moisil-Theodorescu operators.
Note that each operator of the form (\ref{3.1}) corresponds to a
matrix operator with respect to the basis $\{1, e_1, e_2, e_3 \}$.
It has been pointed out in \cite{GS,K1,KS} that some of the most
popular operators of mathematical physics, including Dirac and
Maxwell operators, are of the form (\ref{3.1}).

In this section, we suppose that the coefficients $a_j$ belong to
$SO^1(\sR^3)$ and satisfy
\begin{equation} \label{3.2}
\inf_{x \in \sR^3} |a_j(x)| > 0 \quad \mbox{for} \; j = 1,2,3
\end{equation}
and that the components $\varphi_k$ of $\varphi$ belong to $SO(\sR^3)$.

The main symbol of the operator $A$ is
\[
A_0 (x, \xi) = \sum_{j=1}^3 a_j(x) (i \xi_j) e_j.
\]
Hence,
\[
A_0^2 (x, \xi) = \sum_{j=1}^3 a_j^2(x) \xi_j^2
\]
is a scalar function, and from (\ref{3.2}) we conclude that the 
associated operator $A_0$ is uniformly elliptic on $\sR^3$.
\begin{theorem} \label{t3.1}
The quaternionic operator $A (x, D)$ thought of as acting from
$H^1 (\sR^3, \sH (\sC))$ to $L^2 (\sR^3, \sH (\sC))$ is a Fredholm
operator if and only if
\begin{equation} \label{3.3}
\liminf_{x \to \infty} \left\vert A_0^2 (x, \xi) + \sum_{j=1}^3
\varphi_j(x)^2 \right\vert > 0 \quad \mbox{for every} \; \xi \in 
\sR^3.
\end{equation}
\end{theorem}
\begin{proof}
The limit operators of $A(x, D)$ are the operators with constant
coefficients
\[
A^g(D) u := \sum_{j=1}^3 a_j^g D_{x_j} e_j u + M^{\varphi^g} u.
\]
Let $\check{A}^g (D) := \sum_{j=1}^3 a_j^g D_{x_j} e_j -
M^{\varphi^g}$. Then
\[
A^g (D) \check{A}^g (D) = - \sum_{j=1}^3 (a_j^g)^2 D_{x_j}^2 -
(\varphi^g)^2
\]
where
\[
-(\varphi^g)^2 = (\varphi_1^g)^2 + (\varphi_2^g)^2 +
(\varphi_3^g)^2.
\]
Condition (\ref{3.2}) implies that $A^g (D) : H^1 ( \sR^3, \sH
(\sC)) \to L^2 (\sR^3, \sH (\sC))$ is an invertible operator if
and only if
\begin{equation} \label{3.4}
\inf_{\xi \in \sR^n} \left\vert (A_0^g)^2 (\xi) + \sum_{j=1}^3
(\varphi_j^g)^2 \right\vert > 0 \quad \mbox{for every} \;
\xi \in \sR^3.
\end{equation}
Hence, all limit operators $A^g (D)$ of $A(D)$ are invertible as 
operators from $H^1 (\sR^3, \sH (\sC))$ to $L^2 (\sR^3, \sH (\sC))$ 
if and only if condition (\ref{3.3}) holds.
\end{proof}
\begin{theorem} \label{t3.2}
In addition to the above assumptions, let all functions $a_j$ and
$\varphi_j$ be real-valued, and let $w = \exp v$ be a weight in
$\cR$ with $\lim_{x \to \infty} v(x) = + \infty$. If the condition
\begin{equation} \label{3.5}
\liminf_{x \to \infty} \left(\sum_{j=1}^3 \varphi_j^2(x) -
a_j^2(x) (\frac{\partial v(x)}{\partial x_j})^2 \right)> 0
\end{equation}
is satisfied, then every solution $u \in H^1(\sR^3, \sH (\sC),
w^{-1})$ of the equation $Au = f$ with right-hand side $f \in
L^2(\sR^3, \sH (\sC), w)$ belongs to $H^1(\sR^3, \sH (\sC), w)$.
\end{theorem}
\begin{proof}
Let $t \in [-1, 1]$. The limit operators of $A_{w,t}(x,D)$ are
operators with constant coefficients of the form
\[
A_{w,t}^g (D) = \sum_{j=1}^3 a_j^g (D_{x_j} + it(\frac{\partial
v}{\partial x_j})^g) e_j + M^{\varphi^g}.
\]
As above, let
\[
\check{A}_{w, t}^g (D) := \sum_{j=1}^3 a_j^g (D_{x_j} +
it(\frac{\partial v}{\partial x_j})^g) e_j  - M^{\varphi^g}.
\]
Then $A_{w,t}^g (D) \check{A}_{w,t}^g (D)$ is a scalar operator
with symbol
\[
A_{w,t}^g (\xi) \check{A}_{w,t}^g (\xi) = \sum_{j=1}^3 (a_j^g)^2
(\xi_j + it(\frac{\partial v}{\partial x_j})^g)^2 + \sum_{j=1}^3
(\varphi_j^g)^2(x),
\]
the real part of which is
\[
\mathfrak{R}(A_{w,t}^g (\xi) \check{A}_{w,t}^g (\xi)) =
\sum_{j=1}^3 (a_j^g)^2 \xi_j^2 + \sum_{j=1}^3 (\varphi_j^g)^2 -
t^2[(\frac{\partial v}{\partial x_j})^g]^2.
\]
Condition (\ref{3.5}) implies that
\begin{equation} \label{3.6}
A_{w,t}^g (\xi) \check{A}_{w,t}^g (\xi) \neq 0
\end{equation}
for every $\xi \in \sR^3$ and $t \in [-1, 1]$. Without change of
notation, we now consider $A_{w,t}^g (\xi) $ as a $4 \times 4$
matrix-valued function. The matrix $A_{w,t}^g (\xi)$ is invertible
for every $\xi \in \sR^3$ and $t \in [-1, 1]$ and for every
sequence $g$ which defines a limit operator. Together with
condition (\ref{3.2}), this fact implies that $A_{w,t}^g (D): H^1
(\sR^3, \sH (\sC)) \to L^2 (\sR^3, \sH (\sC))$ is an invertible
operator for every $t \in [-1, 1]$ and for every sequence $g$
which defines a limit operator. Hence, Theorem \ref{t3.2} is a
consequence of Corollary \ref{t1.4}.
\end{proof}
\section{Dirac operators}
\subsection{Essential spectrum of Dirac operators}
In this section we consider the Dirac operator on $\sR^3$ equipped
with the Riemann metric tensor $(\rho_{jk})$ depending on $x \in
\sR^3$ (see for instance \cite{Taller}). We suppose that there is
a constant $C>0$ such that
\begin{equation} \label{4.1}
\rho_{jk}(x) \xi^j \xi^k \geq C |\xi|^2 \quad \mbox{for every} \;
x \in \sR^3
\end{equation}
where we use the Einstein summation convention again. Let $\rho^{jk}$ 
be the tensor inverse to $\rho_{jk}$, and let $\phi^{jk}$ be the 
positive square root of $\rho^{jk}$. The Dirac operator on $\sR^3$ is 
the operator
\begin{equation} \label{4.2}
\cD := \frac{c}{2} \gamma_k (\phi^{jk} P_j + P_j \phi^{jk}) +
\gamma_0 c^2 m - e \Phi
\end{equation}
acting on functions on $\sR^3$ with values in $\sC^4$. In
(\ref{4.2}), the $\gamma_k$, $k = 0, 1, 2, 3$, are the $4 \times
4$ Dirac matrices, i.e., they satisfy
\begin{equation} \label{4.3}
\gamma_j \gamma_k + \gamma_k \gamma_j = 2 \delta_{jk}E
\end{equation}
for all choices of $j, k = 0, 1, 2, 3$ where $E$ is the $4 \times
4$ unit matrix,
\[
P_j = D_j + \frac{e}{c} A_j, \qquad D_j = \frac{h}{i}
\frac{\partial}{\partial x_j}, \quad j=1,2,3,
\]
where $h$ is the Planck constant, $\vec{A} = (A_1, A_2, A_3)$ is
the vector potential of the magnetic field $\vec{H}$, that is
$\vec{H} = \rot \vec{A}$, $\Phi$ is the scalar potential of the
electric field $\vec{E}$, that is $\vec{E} = \grad \Phi$, and $m$
and $e$ are the mass and the charge of the electron.

We suppose that $\rho^{jk}, \, A_j$ and $\Phi$ are real-valued
functions which satisfy the conditions
\begin{equation} \label{1.2.3.}
\rho^{jk} \in SO^1(\sR^3), \quad A_j \in SO^1(\sR^3), \quad \Phi
\in SO(\sR^3)
\end{equation}
for $j, k =1, 2, 3$. We consider the operator $\cD$ as an
unbounded operator on the Hilbert space $L^2(\sR^3, \sC^4)$ with
domain $H^1(\sR^3, \sC^4)$. The conditions imposed on the magnetic
and electric potentials $\vec{A}$ and $\Phi$ guarantee the
self-adjointness of $\cD$. The main symbol of $\cD$ is
\[
a_0 (x, \xi) = c \phi^{jk}(x) \xi_j \gamma_k.
\]
Using (\ref{4.3}) and the identity $\phi^{jk}(x) \phi^{rt}(x) 
\delta_{kt} = \rho^{jr}(x)$ we obtain
\begin{eqnarray} \label{4.4}
a_0 (x, \xi)^2 & = & c^2 h^2 \phi^{jk}(x) \phi^{rt}(x) \xi_j \xi_r
\gamma_k \gamma_t \nonumber \\
& = & c^2 h^2 \phi^{jk}(x) \phi^{rt}(x) \delta_{kt} \xi_j \xi_r
\nonumber \\
& = & c^2 h^2 \rho^{jr}(x) \xi_j \xi_r E.
\end{eqnarray}
Together with (\ref{4.1}) this equality shows that $\cD$ is a
uniformly elliptic differential operator on $\sR^3$.

Conditions (\ref{1.2.3.}) imply that limit operators $\cD_g$ of
$\cD$ defined by the sequences $g : \sN \to \sZ^3$ tending to
infinity are operators with constant coefficients of the form
\begin{equation} \label{4.5}
\cD_g = c \gamma_k \phi_q^{jk} (D_j + \frac{e}{c} A_j^g) +
\gamma_0 m c^2 - e \Phi^g
\end{equation}
where
\begin{equation} \label{4.6}
\phi_g^{jk} := \lim_{m \to \infty} \phi^{jk}(g(m)), \quad A_j^g :=
\lim_{m \to \infty} A_j(g(m)), \quad \Phi^g := \lim_{m \to \infty}
\Phi (g(m)).
\end{equation}
Note that one can suppose without loss of generality that the
sequence $g$ is such that the limits in (\ref{4.6}) exist. In the
opposite case we pass to a suitable subsequence of $g$.

The operator $\cD_g$ is unitarily equivalent to the operator
\[
\cD_g^1 = c \gamma_l \omega_g^{jl} D_j + \gamma_0 m c^2 - e
\Phi^g,
\]
and the equivalence is realized by the unitary operator
\[
T_{\vec{A}^g} : f \mapsto e^{i \frac{e}{c} \vec{A}^g \cdot x} f
\quad \mbox{where} \quad \vec{A}^g := (A_1^g, A_2^g, A_3^g).
\]
Let
\[
\Phi^{\sup} := \limsup_{x \to \infty} \Phi (x), \quad \Phi^{\inf}
:= \liminf_{x \to \infty} \Phi (x).
\]
Then the interval $[\Phi^{\inf}, \, \Phi^{\sup}]$ is just the
set of all partial limits $\Phi^g$ of function $\Phi$ as $x \to
\infty$.
\begin{theorem} \label{te4.1}
Let conditions $(\ref{1.2.3.})$ be fulfilled. Then the Dirac
operator
\[
\cD : H^1(\sR^3, \sC^4) \to L^2(\sR^3, \sC^4)
\]
is a Fredholm operator if and only if
\begin{equation} \label{4.7}
[\Phi^{\inf}, \, \Phi^{\sup}] \subset (- mc^2/e, \, mc^2/e).
\end{equation}
\end{theorem}
\begin{proof}
Set $\cD_0^g (\xi) := ch \gamma_k \phi_g^{jk} \xi_j + \gamma_0 mc^2$
and $\rho _g^{jk} := \lim_{m\to \infty }\rho ^{jk}(g_m)$. Then
\begin{eqnarray} \label{4.8}
\lefteqn{(\cD_{0}^g (\xi) - e \Phi_1^g E) \, (\cD_{0}^g (\xi) +
e \Phi^g E)} \nonumber \\
&& = \left( c^2 h^2 \rho_g^{jk} \xi_j \xi_k + m^2 c^4 - (e \Phi^g)
^2 \right) E.
\end{eqnarray}
Let condition (\ref{4.7}) be fulfilled. Then every partial limit
$\Phi^g = \lim_{k \to \infty} \Phi (g(k))$ of $\Phi$ lies in the
interval $(- mc^2/e, \, mc^2/e)$. The identity (\ref{4.8}) implies
that
\[
\det (\cD_0^g(\xi) - e \Phi^g E) \neq 0
\]
for every $\xi \in \sR^3$. Hence, the operator $\cD_g^1 :
H^1(\sR^3, \sC^4) \to  L^2(\sR^3, \sC^4)$ is invertible and,
consequently, so is $\cD_g$. By Theorem \ref{t1.1}, $\cD$ is a
Fredholm operator.

For the reverse implication, assume that condition (\ref{4.7}) is
not fulfilled. Then there exist a number $\Phi^g \notin (-mc^2/e,
\, mc^2/e)$ and a vector $\xi^0 \in \sR^3$ such that
\[
c^2g_g^{jk} \xi_j^0 \xi_k^0 + m^2 c^4 - (e \Phi^g)^2 = 0.
\]
Given $\xi^0$ we find a vector $u \in \sC^4$ such that $v :=
\left( \cD_0^g (\xi^0) + (e \Phi^g) E\right) u \neq 0$. Then
(\ref{4.8}) implies that
\[
\left( \cD_0^g (\xi^0) - e \Phi^g E \right) v = 0,
\]
whence
\[
\det (\cD_0^g (\xi^0) - e\Phi^g E)=0.
\]
Thus, the operator $\cD_g$ is not invertible. By Theorem
\ref{t1.1}, $\cD$ cannot be a Fredholm operator.
\end{proof} 
\begin{theorem} \label{te4.2}
If condition $(\ref{1.2.3.})$ is satisfied, then
\[
\spe \cD = (-\infty, -e \Phi^{\inf} - mc^2] \cup [-e \Phi^{\sup} +
mc^2, +\infty).
\]
\end{theorem}
\begin{proof}
Let $\lambda \in \sC$. The symbol of the operator $\cD_g - \lambda
I$ is the function $\xi \mapsto \cD_0^g(\xi) - (e \Phi^g +
\lambda) E$. Invoking (\ref{4.8}) we obtain
\begin{eqnarray} \label{4.10}
\lefteqn{(\cD_0^g (\xi) - (e \Phi^g + \lambda) E) \, (\cD_0^g(\xi)
+ (e \Phi^g + \lambda)E)} \nonumber \\
&& = \left( c^2 \rho_g^{jk} \xi_j \xi_k + m^2c^4 - (e \Phi^g +
\lambda)^2 \right) E.
\end{eqnarray}
Repeating the arguments from the proof of Theorem \ref{te4.1}, we
find that the eigenvalues $\lambda_\pm^g (\xi)$ of the matrix
$\cD_0^g (\xi) - e \Phi_1^gE$ are given by
\begin{equation} \label{4.11}
\lambda_\pm^g (\xi) := - e \Phi^g \pm (c^2 \rho_g^{jk} \xi_j \xi_k
+ m^2 c^4)^{1/2}.
\end{equation}
From (\ref{4.11}) we further conclude
\begin{eqnarray*}
\{ \lambda \in \sR : \lambda = \lambda_-^q (\xi), \, \xi \in \sR^3
\} & = & (-\infty, -e \Phi^g - mc^2], \\
\{ \lambda \in \sR : \lambda = \lambda_+^q (\xi), \, \xi \in \sR^3
\} & = & [-e \Phi^g + mc^2, +\infty).
\end{eqnarray*}
Hence,
\[
\spec \cD^g = (-\infty, -e \Phi^g - mc^2] \cup [-e \Phi^g + mc^2,
+\infty),
\]
whence the assertion via Theorem \ref{t1.2}.
\end{proof} \\[3mm]
Thus, if $\Phi^{\sup} - \Phi^{\inf} \ge 2 m c^2/e$, then $\spe
\cD$ is all of $\sR$, whereas $\spe \cD$ has a proper gap in the
opposite case.
\subsection{Exponential estimates of eigenfunctions of the Dirac
operator}
\begin{theorem} \label{te4.4}
Let the conditions $(\ref{1.2.3.})$ be fulfilled. Let $\lambda$ be
an eigenvalue of $\cD$ which lies in the gap $(-e \Phi^{\inf} -
mc^2, \, -e \Phi^{\sup} + mc^2)$ of the essential spectrum.
Further, let $w = \exp v$ be a weight in $\cR$ with $\lim_{x \to
\infty} w(x) = \infty$ which satisfies
\begin{equation} \label{e4.19}
\limsup_{x \to \infty} |\nabla v(x)|_{\rho (x)} < \frac{1}{ch}
\sqrt{m^2 c^4 - (e \Phi^{\sup} + \lambda)^2}.
\end{equation}
Then every eigenfunction of $\cD$ associated with $\lambda$
belongs to $H^1(\sR^3,\sC^4, w)$.
\end{theorem}
\begin{proof}
Let $\lambda \in (-e \Phi^{\inf} - mc^2, -e \Phi^{\sup} + mc^2)$
be an eigenvalue of $\cD$. As above, we examine the spectra of the
limit operators $(\cD_{w,t})_g$ of $\cD_{w,t} := w^{-t} \cD w^t$
for $t$ running through $[0,1]$. Let $(\cD_{w,t})_g$ be a limit
operator of $\cD_{w,t}$ with respect to a sequence $g$ tending to
infinity. One easily checks that $(\cD_{tw})_g$ is unitarily
equivalent to the operator
\[
(\cD_{tw}^\prime)_g := A_{t,g} - e \Phi^g E
\]
where
\[
A_{t,g} := c \gamma_k \phi_g^{jk} (D_j + ith (\frac{\partial v}
{\partial x_j})^g) + \gamma_0 mc^2.
\]
The operator $A_{t,g}$ has constant coefficients, and its symbol
is
\[
\widehat{A_{t,g}}(\xi) = c \gamma_k \phi_g^{jk} (h(\xi_j + ith
(\frac{\partial v}{\partial x_j})^g)) + \gamma_0 mc^2.
\]
Further,
\begin{eqnarray*}
\lefteqn{\mathfrak{R} \left[ \left( \widehat{A_{t,g}}(\xi) - (e
\Phi^g - \lambda)E \right) \left( \widehat{A_{t,g}}(\xi) + (e
\Phi^g - \lambda)E \right) \right]} \\
&& = \mathfrak{R} \left[ c^2 h^2 \rho_g^{jk} \left(\xi_j + ith
(\frac{\partial v}{\partial x_j})^g \right) \left(\xi_k + ith
(\frac{\partial v}{\partial x_k})^g \right) \right] \\
&& \qquad + \; \mathfrak{R} \left[ \left( m^2 c^4 - (e \Phi^g +
\lambda)^2 \right) E \right] \\
&& = \left[ c^2 h^2 \rho_g^{jk} \xi_j \xi_k - c^2 h^2 t^2
\rho_g^{jk} (\frac{\partial v}{\partial x_j})^g (\frac{\partial
v}{\partial x_k})^g + \left( m^2 c^4 - (e \Phi^g + \lambda)^2
\right) \right] E  \\
&& =: \gamma _{g,t}(\xi, \lambda) E.
\end{eqnarray*}
Assume that condition (\ref{e4.19}) is fulfilled. Then, since $c^2
h^2 \rho_g^{jk} \xi_j \xi_k \ge 0$,
\[
\inf_{\xi \in \sR^n} \gamma_{g,t} (\xi, \lambda )>0
\]
for all $t \in [0,1]$ and for all sequences $g \to \infty$ for
which the limit operators exist. Hence, (\ref{e4.19}) implies that
the matrix $\widehat{A_{t,g}}(\xi) - (e \Phi^g - \lambda)E$ is
invertible for every $\xi \in \sR^3$. On the other hand, due to
the uniform ellipticity of $A_{t,g}$ one has $\lambda \in \spec
\left( \widehat{A_{t,g}}(\xi) - (e \Phi^g - \lambda)E \right)$ if
and only if there exists a $\xi_0 \in \sR^3$ such that the matrix
$\widehat{A_{t,g}}(\xi_0) - (e \Phi^g + \lambda)E$ is not
invertible. Thus, $\lambda \notin \spec (\cD_{tw})_g$ for every $t
\in [0,1]$ and every sequence $g \to \infty$. Via Corollary
\ref{t1.4}, the assertion follows.
\end{proof} \\[3mm]
We conclude by an example. Let the conditions (\ref{1.2.3.}) be
fulfilled, and let $\lambda$ be an eigenvalue of $\cD$ in $(-e
\Phi^{\inf} - mc^2, -e \Phi^{\sup} + mc^2)$ and $u_\lambda$ an
associated eigenfunction. If $a$ satisfies the estimates
\[
0 < a < \frac{\sqrt{m^2 c^4 - (e \Phi^{\sup} + \lambda)^2}} {ch
\rho^{\sup}}
\]
where
\[
\rho^{\sup} := \liminf_{x \to \infty} \sup_{\omega \in S^2}
(\rho^{jk}(x) \omega_j \omega_k)^{1/2},
\]
then $u_\lambda \in H^1(\sR^3, \sC^4, e^{a \langle x \rangle})$.

\end{document}